\documentclass[runningheads,envcountsame]{llncs}
\usepackage{amssymb,amsfonts,amsmath}
\usepackage{booktabs,tabularx}
\usepackage{url,xspace,soul,wrapfig}
\usepackage{graphicx,hyperref,paralist}
\usepackage{paralist,enumerate}
\usepackage[textsize=footnotesize]{todonotes}
\usepackage{enumerate}
\usepackage[capitalise]{cleveref}
\usepackage{subcaption}
\usepackage{timetravel}
\usepackage[inline]{enumitem}
\usepackage{tabularx}
\usepackage[misc,geometry]{ifsym} 
\setCounters{theorem}

\newif\ifproc
\procfalse

\newcommand{\bind}[1]{\ensuremath{#1_{\operatorname{bind}}}}

\let\doendproof\endproof
\renewcommand\endproof{~\hfill$\qed$\doendproof}

\makeatother

\graphicspath{{figs/}}

\title{Crossing Numbers of Beyond-Planar Graphs%
	\thanks{Research in this work started at the Bertinoro Workshop on Graph Drawing 2019.
		MC was supported by DFG under grant CH 897/2-2.
		FM was supported in part by MIUR under grant 20174LF3T8 AHeAD: efficient Algorithms for HArnessing networked Data.}}
\titlerunning{Crossing Numbers of Beyond-Planar Drawings}

\author{Markus~Chimani\inst{1}\ifproc\orcidID{0000-0002-4681-5550}\fi
\and 
Philipp~Kindermann\inst{2}\ifproc\orcidID{0000-0001-5764-7719}\fi
\and 
Fabrizio~Montecchiani\inst{3}\ifproc\orcidID{0000-0002-0543-8912}\fi
\and 
Pavel~Valtr\inst{4}}

\authorrunning{M. Chimani et al.}

\institute{
Osnabr\"uck University, Osnabrück, Germany\\
\email{markus.chimani@uos.de}
\and
University of W\"urzburg, W\"urzburg, Germany\\
\email{philipp.kindermann@uni-wuerzburg.de}
\and
University of Perugia, Perugia, Italy\\
\email{fabrizio.montecchiani@unipg.it}
\and
Charles University in Prague, Prague, Czech Republic\\
\email{valtr@kam.mff.cuni.cz}
}

\newcommand{\crn}{\ensuremath{\mathrm{cr}}}
\newcommand{\op}{\text{$1$-pl}}
\newcommand{\kp}{\text{$k$-pl}}
\newcommand{\fp}{\text{fan}}
\newcommand{\qp}{\text{quasi}}

\newcommand{\ratio}{crossing ratio\xspace}
\newcommand{\rr}{\ensuremath{\varrho}}
\newcommand{\bigo}{\ensuremath{\mathcal{O}}}

\begin{document}
\maketitle

\pagenumbering{arabic}

\begin{abstract} 
We study the $1$-planar, quasi-planar, and fan-planar crossing number in comparison to the (unrestricted) crossing number of graphs. We prove that there are $n$-vertex $1$-planar (quasi-planar, fan-planar) graphs such that any $1$-planar (quasi-planar, fan-planar) drawing has $\Omega(n)$ crossings, while $\bigo(1)$ crossings suffice in a crossing-minimal drawing without restrictions on local edge crossing patterns.
\end{abstract}

\section{Introduction}
The \emph{crossing number} of a graph $G$, denoted by $\crn(G)$, is the smallest number of pairwise edge crossings over all
possible drawings of $G$. Many papers are devoted to the study of this parameter, refer to~\cite{schaefer,vrto} for surveys. In particular, minimizing the number of crossings is one of the seminal problems in graph drawing (see, e.g.,~\cite{DBLP:conf/er/BatiniFN85,DBLP:journals/tse/BatiniNT86,DBLP:journals/tsmc/SugiyamaTT81}), whose importance has been further witnessed by user studies showing how edge crossings may deteriorate the readability of a diagram~\cite{DBLP:journals/iwc/Purchase00,DBLP:journals/ese/PurchaseCA02,DBLP:journals/ivs/WarePCM02}. On the other hand, determining the crossing number of a graph is NP-hard~\cite{DBLP:journals/dcg/Bienstock91} and can be solved exactly only on small/medium instances~\cite{exactcrmin}. %
 On the positive side, the crossing number is fixed-parameter tractable in the number of crossings~\cite{DBLP:conf/stoc/KawarabayashiR07} and can be approximated by a constant factor for graphs of bounded degree and genus~\cite{DBLP:conf/soda/HlinenyC10}.

A recent research stream studies graph drawings where, rather than minimizing the number of crossings, some edge crossing patters are forbidden; refer to~\cite{DBLP:journals/jgaa/BekosKM18,DBLP:journals/csur/DidimoLM19,DBLP:journals/dagstuhl-reports/Hong0KP16,DBLP:journals/shonan-reports/HongT216} for surveys and reports. A key motivation for the study of so-called \emph{beyond-planar graphs} are recent cognitive experiments showing that already the absence of specific kinds of edge crossing configurations has a positive impact on the human understanding of a graph drawing~\cite{DBLP:journals/vlc/HuangEH14,DBLP:journals/siamjo/Mutzel01}. Of particular interest for us are three families of beyond-planar graphs that have been extensively studied, namely the $k$-planar, fan-planar, and $k$-quasi-planar graphs; refer to~\cite{DBLP:journals/csur/DidimoLM19} for additional families. A \emph{$k$-planar drawing} is such that each edge is crossed at most $k \ge 1$ times~\cite{DBLP:journals/combinatorica/PachT97} (see also~\cite{DBLP:journals/csr/KobourovLM17} for a survey on $1$-planarity).  A \emph{$k$-quasi planar} drawing does not have $k \ge 3$ mutually crossing edges~\cite{DBLP:journals/dcg/AlonE89}. A \emph{fan-planar drawing} does not contain two independent edges that cross a third one or two adjacent edges that cross another edge from different ``sides''~\cite{DBLP:journals/corr/KaufmannU14}. A graph is \emph{$k$-planar} (\emph{$k$-quasi-planar}, \emph{fan-planar}) if it admits a $k$-planar ($k$-quasi-planar, fan-planar) drawing; a $3$-quasi-planar graph is simply called \emph{quasi-planar}.

\begin{table}[t]
  \centering
  \caption{Lower and upper bounds the crossing ratio of beyond-planar graphs.}
  \label{tab:contribution}
  \medskip
  \begin{tabular}{l@{\qquad}l@{\qquad}l}
    \toprule
    Graph class & lower bound & upper bound \\
    \hline
    1-planar & $n/2-1$ & $n/2-1$ \\
    quasi-planar & $\Omega(n)$ & $O(n^2)$ \\
    $k$-quasi-planar & $\Omega(n/k^3)$ & $f(k)\cdot n^2\log^2 n$ \\
    fan-planar & $\Omega(n)$ & $O(n^2)$\\
    \bottomrule
  \end{tabular}
\end{table}

In this context, an intriguing question is to what extent edge crossings can be minimized while forbidding such local crossing patterns. In particular, we ask whether avoiding local crossing patterns in a drawing of a graph may enforce an overall large number of crossings, whereas only a few crossings would suffice in a crossing-minimal drawing of the graph. %
We answer this question in the affirmative for the above-mentioned three families of beyond-planar graphs. Our contribution are summarized in Table~\ref{tab:contribution}.

\begin{enumerate}[nosep]
\item In~\cref{sse:1pl}, we prove that there exist $n$-vertex $1$-planar graphs such that the ratio between the minimum number of crossings in a $1$-planar drawing of one such graph and its crossing number is $n/2-1$. This result can be easily extended to $k$-planar graphs if we allow parallel edges.

\item In~\cref{sse:qp}, we prove that there exist $n$-vertex quasi-planar graphs such that the ratio between the minimum number of crossings in a quasi-planar drawing of one such graph and its crossing number is $\Omega(n)$. Similarly, a $\Omega(n/k^3)$ bound can be proved for $k$-quasi-planar graphs.

\item In~\cref{sse:fp}, we prove that there exist $n$-vertex fan-planar graphs such that the ratio between the minimum number of crossings in a fan-planar drawing of one such graph and its crossing number is $\Omega(n)$.
\end{enumerate}

\noindent The lower bound in Result 1 is tight. Since fan-planar and quasi-planar graphs have $\bigo(n)$ edges, the lower bounds in Results 2 and 3 are a linear factor from the trivial upper bound $\bigo(n^2)$, and it remains open whether such an upper bound can be achieved (see~\cref{se:open}). All results are based on nontrivial constructions that exhibit interesting structural properties of the investigated graphs. 

\paragraph{Notation and Definitions.}\label{sse:def}

We assume familiarity with standard definitions about graph drawings and embeddings of planar and nonplanar graphs (see, e.g.,~\cite{dett-gdavg-99,DBLP:journals/csur/DidimoLM19}). In a drawing of a graph, we assume that an edge does not contain a vertex other than its endpoints, no two edges meet tangentially, and no three edges share a crossing. It suffices to only consider {\em simple} drawings  where any two edges intersect in at most one point, which is either a common endpoint or an interior point where the two edges properly cross. Thus, in a simple drawing, any two adjacent edges do not cross and any two non-adjacent edges cross at most once.

We define the \emph{$k$-planar crossing number} of a $k$-planar graph $G$, denoted by $\crn_\kp(G)$, as the 
minimum number of crossings over all $k$-planar drawings of $G$. 
The \emph{$k$-planar \ratio}~$\rr_\kp$ is the supremum of $\crn_\kp(G)/\crn(G)$ 
over all $k$-planar graphs $G$. Analogously, we define 
the \emph{quasi-planar} and the \emph{fan-planar crossing number} 
of a graph $G$, denoted by $\crn_\qp(G)$ and $\crn_\fp(G)$, as well as the 
\emph{quasi-planar}  and the  \emph{fan-planar \ratio}, denoted by 
$\rr_\qp$ and~$\rr_\fp$.

\section{The $1$-planar \ratio}\label{sse:1pl}

An $n$-vertex $1$-planar graph has at most $4n-8$ edges and a $1$-planar drawing has at most $n-2$ crossings, that is $\crn_\op(G) \le n-2$~\cite{DBLP:journals/csr/KobourovLM17}. Observe that for $\crn(G)  < \crn_\op(G)$ it has to hold that $\crn(G) \ge 2$. It follows that the $1$-planar \ratio is $\rr_\op \le  n/2-1$. We show that this bound can be achieved.

\begin{theorem}\label{thm:1planar}
For every $\ell \ge 7$, there exists a $1$-planar graph $G_\ell$ with $n=11\ell+2$ vertices such that $\crn_\op(G_\ell)=n-2$ and $\crn(G_\ell)=2$, which yields the largest possible $1$-planar \ratio.
\end{theorem}

The construction of~$G_\ell$ consists of three parts: a rigid graph~$P$ that has to be drawn planar in any 1-planar drawing; its dual~$P^*$; a set of \emph{binding} edges and one \emph{special} edge that force~$P$ and~$P^*$ to be intertwined in any 1-planar drawing.

To obtain~$P$, we utilize a construction introduced by Korzhik and Mohar~\cite{km-mo1ih-jgt13}. 
They construct graphs~$H_\ell$ that are the medial extension of the Cartesian product of the path of length~2 and the cycle of length~$\ell$; see \cref{fi:1-planar-P}. 
They prove that~$H_\ell$ has exactly one 1-planar embedding on the sphere, and that embedding is crossing-free. 
We choose $P=H_{\ell}$ as our rigid graph and fix its (1-)planar embedding (when we will refer to~$P$, we will usually mean this embedding). 

Let~$P^*$ be the dual of~$P$, obtained by placing a dual vertex $h^*$ into each face~$h$ of~$P$ and connecting two dual vertices if their corresponding faces share an edge; see \cref{fi:1-planar-P-dual}. 
Since~$P$ has~$5\ell$ vertices and $11\ell$ edges, by Euler's polyhedra formula it has $6\ell+2$ faces; thus,~$P^*$ has $6\ell+2$ vertices and~$11\ell$ edges.

Obviously, $P\cup P^*$ can be drawn planar, as both~$P$ and~$P^*$ are planar and disjoint. 
All faces of~$P$ have size~3 or~4, except two large (called \emph{polar}) faces~$f$ and~$g$ of size~$\ell$.
We create a graph~$G'$ by adding~$\ell$ \emph{binding} edges to~$P\cup P^*$ between~$f^*$ (the vertex of~$P^*$ corresponding to face~$f$) and the vertices of~$P$ that are incident to~$f$.
This forces~$f^*$ to be drawn in face~$f$ in any 1-planar drawing.
\ifproc
In the full version~\cite{fullversion}
\else
In~\cref{app:1pl},
\fi 
we prove the 
following lemma, cf.\ \cref{fi:1-planar-1-planar} and \cref{fi:1-planar-2-crossings}.

\begin{figure}[t]
\begin{subfigure}[t]{.47\columnwidth}
  \includegraphics[page=21]{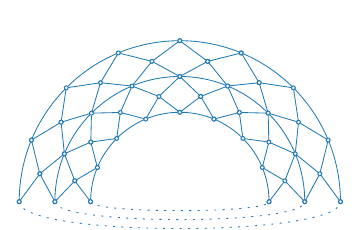}
  \caption{The graph $P$}
  \label{fi:1-planar-P}
\end{subfigure}
\hfill
\begin{subfigure}[t]{.47\columnwidth}
  \includegraphics[page=24]{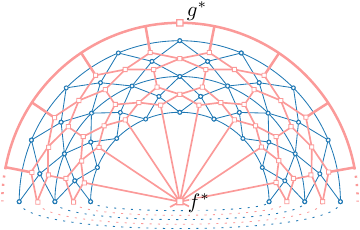}
  \caption{The graph $P\cup P^*$}
  \label{fi:1-planar-P-dual}
\end{subfigure}

\begin{subfigure}[t]{.47\columnwidth}
  \includegraphics[page=26]{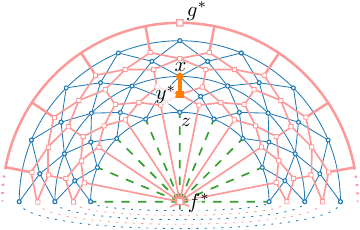}
  \caption{$\crn_\op(G_\ell)\le n-2$}
  \label{fi:1-planar-1-planar}
\end{subfigure}
\hfill
\begin{subfigure}[t]{.47\columnwidth}
  \includegraphics[page=28]{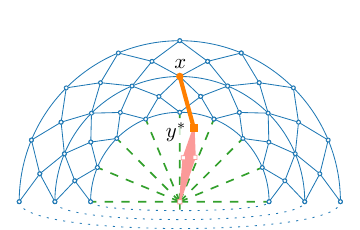}
  \caption{$\crn(G_\ell)\le 2$}
  \label{fi:1-planar-2-crossings}
\end{subfigure}
\caption{Construction of the graph~$G_\ell$ in the proof of \cref{thm:1planar}. Blue circles and edges are~$P$; red squares and bold edges are~$P^*$;
green dashed edges are the binding edges; and the orange very bold edge is the special edge.}
\end{figure}

\newcommand{\lemOPlanarTDrawings}{%
  $G'$ has only two types of 1-planar embeddings (up to the choice of the outer face): 
  a planar one where~$P^*$ lies completely inside face~$f$ of~$P$; and
  a 1-planar embedding where~$f^*$ lies inside~$f$, $g^*$ lies inside~$g$, and each
  edge of~$P$ crosses an edge of~$P^*$ and vice versa.}
\wormhole{1-planar-2-drawings}
\begin{lemma}\label{lem:1-planar-2-drawings}
  \lemOPlanarTDrawings
\end{lemma}

Let~$z$ be a vertex of~$P$ on the boundary of~$f$. Let~$y$ be the face of size~4
that has~$z$ on its boundary. Let~$x$ be the degree-6 vertex on the boundary of~$y$.
We obtain~$G_\ell$ from~$G'$ by adding the \emph{special} edge $(x,y^*)$.
In the planar embedding of \cref{lem:1-planar-2-drawings},
$P^*$ and thus $y^*$ lies inside 
face~$f$ of~$P$, so $(x,y^*)$ has to cross at least two edges of~$P$; 
see \cref{fi:1-planar-2-crossings}.
Choosing the face that corresponds to~$z$ as the outer face of $P^*$
gives a non-1-planar drawing of~$G_\ell$ with~2 crossings.

Hence,~$G'$ has to be drawn in the second way of \cref{lem:1-planar-2-drawings};
see \cref{fi:1-planar-1-planar}. 
Here, the edge $(x,y^*)$ can be added without further crossings. 
Graph~$G_\ell$ consists of %
$n=11\ell+2$ vertices in total. 
Both $P$ and~$P^*$ have $11\ell$ edges, and each of them is crossed, so there are $n-2$ crossings in total, which is the maximum possible in a 1-planar drawing.
Hence, $\crn_\op(G_\ell)=n-2$ and $\crn(G_\ell)=2$, so $\rr_\op \le  n/2-1$.

\smallskip

The construction used in the proof of \cref{thm:1planar} can be generalized to $k$-planar multigraphs. It suffices to replace each edge of $G_\ell$, except the special edge, by a bundle of $k$ parallel edges: 

\begin{corollary}\label{co:kplanar}
For every $\ell \ge 6$, there exists a $k$-planar multigraph $G_{\ell,k}$ with $n=11\ell+2$ vertices and maximum edge multiplicity $k$ such that $\crn_\kp(G_{\ell,k})=k^2\,(n-2)$ and $\crn(G_{\ell,k})=2k$, thus $\rr_\kp \ge k\,(n-2)/2$.
\end{corollary}

\section{The quasi-planar \ratio}\label{sse:qp}

An $n$-vertex quasi-planar graph $G$ has at most $6.5n-20$ edges, thus $\crn_\qp(G) \in \bigo(n^2)$~\cite{DBLP:journals/csur/DidimoLM19}. 
For $\crn(G)  < \crn_\qp(G)$ it has to hold that $\crn(G) \ge 2$, and hence $\rr_\qp \in \bigo(n^2)$. 
We show that the quasi-planar \ratio is unbounded, even for~$\crn(G)\le 3$:

\begin{theorem}\label{thm:quasi-planar-linear}
For every $\ell \ge 2$, there exists a quasi-planar graph $G_\ell$ with $n=12\ell-5$ vertices such that $\crn_\qp(G_\ell) \ge \ell$ and $\crn(G_\ell) \le 3$, thus $\rr_\qp \in \Omega(n)$.
\end{theorem}

In order to prove \cref{thm:quasi-planar-linear}, we begin with a technical lemma.

\begin{lemma}\label{lem:parallel}
Let $G$ be a graph containing two independent edges $(u,v)$ and $(w,z)$. 
Suppose that $u$ and $v$ ($w$ and $z$, resp.) are connected by a set $\Pi_{uv}$ ($\Pi_{wz}$, resp.) of $\ell-1$ paths of length two.
Let $\Gamma$ be a drawing of $G$. 
If $(u,v)$ and $(w,z)$ cross in $\Gamma$, then $\Gamma$ contains at least $\ell$ crossings.
\end{lemma}
\begin{proof}
Suppose that $(u,v)$ and $(w,z)$ cross. 
If each of the $\ell-1$ paths in $\Pi_{wz}$ crosses $(u,v)$, then the claim follows. 
Assume otherwise that at least one of these paths does not cross $(u,v)$. 
This path forms a $3$-cycle $t$ with $(w,z)$; the $\ell-1$ paths of $\Pi_{uv}$ all cross at least one edge of $t$, which proves the claim.
\end{proof}

\begin{figure}[t]
\centering
\subcaptionbox{$\crn(G_\ell) \le 3$\label{fi:quasiplanar-1}}{\includegraphics[page=1]{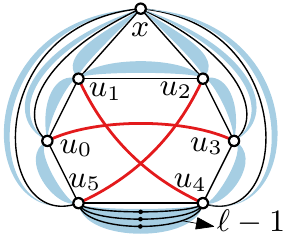}}
\hfill
\subcaptionbox{$\crn_\qp(G_\ell) {\le} 2\ell{+}1$\label{fi:quasiplanar-2}}{\includegraphics[page=2]{quasiplanar-linear}~}
\hfill
\subcaptionbox{$C$ is not crossed\label{fi:quasiplanar-3}}{\includegraphics[page=3]{quasiplanar-linear}}
\hfill
\subcaptionbox{$C$ is crossed\label{fi:quasiplanar-4}}{\includegraphics[page=4]{quasiplanar-linear}}
\caption{Illustration for the proof of \cref{thm:quasi-planar-linear}.}
\end{figure}

\begin{proof}[of \cref{thm:quasi-planar-linear}]
Let $G_\ell$ be the graph constructed as follows; cf.\ \cref{fi:quasiplanar-1}. 
Start with a 6-cycle $C=\langle u_0,u_1,\dots,u_5 \rangle$, and a vertex $x$ connected to each of $C$, yielding graph $G'$.
\emph{Extend} each edge of $G'$ by adding $\ell-1$ disjoint paths of length two between its endpoints.
Finally, add \emph{special} edges $(u_i,u_{i+3})$, $i=0,1,2$.

The resulting graph $G_\ell$ has $n=12(\ell-1)+7=12\ell-5$ vertices and admits a drawing with~3 crossings, so $\crn(G_\ell) \le 3$; see \cref{fi:quasiplanar-1}. 
Note that $G_\ell$ admits a quasi-planar drawing with $2\ell+1$ crossings as shown in \cref{fi:quasiplanar-2}.
We prove that $\crn_\qp(G_\ell) \ge \ell$.  Let $\Gamma$ be a quasi-planar drawing of $G_\ell$. 
If there are two edges of~$G'$ that cross each other, then the claim follows by \cref{lem:parallel}. %

If no special edge would cross $G'$, they would all
be drawn within the unique face of size 6 in $G'$.
They would mutually cross, contradicting quasi-planarity.

Thus, at least one special edge, say $s=(u_0,u_3)$, crosses an edge $(a,b)$ 
of~$G'$.
Consider the closed (possibly self-intersecting) curve $\mathcal L$ 
composed of $s$ plus the subpath of 
$C$ connecting $u_0$ to $u_3$ and containing none of the vertices $a$ and~$b$. 
This curve partitions the plane into two or more regions, and $a$ 
and $b$ lie in different regions; see~\cref{fi:quasiplanar-3}--\cref{fi:quasiplanar-4} for an illustration.  
Thus $(a,b)$ and the $\ell-1$ paths connecting $a$ and $b$  cross $\mathcal L$, yielding $\ell$ crossings in $\Gamma$, as desired. 
\end{proof}

The above proof can be straight-forwardly extended to $k$-quasi-planar graphs by using exactly the same construction in which the cycle $C$ has length $2k$.
Note that any $k$-quasi-planar graph has at most $c_k n\log n$ edges, where~$c_k$ depends only on~$k$~\cite{DBLP:journals/comgeo/SukW15},
so $\rr_\qp\le f(k)\cdot n^2\log^2 n$.

\begin{corollary}
  \label{co:k-quasi-planar}
  For every $\ell \ge 2$ and $k \ge 3$, there exists a $k$-quasi-planar 
  graph~$G_{\ell,k}$ with $n=2k(\ell+1)+1$ vertices such 
  that $\crn_\qp(G_{\ell,k}) \ge \ell$ 
  and $\crn(G_{\ell,k}) \le k(k-1)/2$, 
  thus $\rr_\qp \in \Omega(n/k^3)$.
\end{corollary}

\section{The fan-planar \ratio}\label{sse:fp}

An $n$-vertex fan-planar graph $G$ has at most $5n-10$ edges, thus $\crn_\fp(G) \in \bigo(n^2)$~\cite{DBLP:journals/csur/DidimoLM19}. For $\crn(G)  < \crn_\fp(G)$ it has to hold that $\crn(G) \ge 2$, and hence $\rr_\fp \in \bigo(n^2)$. We show that the fan-planar \ratio is unbounded, even for $\crn(G)=3$. 

\begin{theorem}\label{thm:fan-planar-linear}
For every $\ell \ge 2$, there exists a fan-planar graph $G_\ell$ with $n=9\ell+1$ vertices such that $\crn_\fp(G_\ell) = \ell$ and $\crn(G_\ell)=3$, thus $\rr_\fp \in \Omega(n)$.
\end{theorem}

\begin{figure}[t]
\centering
\subcaptionbox{$\crn(G_\ell) \le 2$.\label{fi:fanplanar-1}}{\includegraphics[page=1]{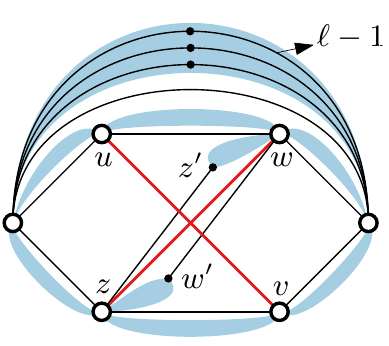}}
\hfil
\subcaptionbox{$\crn_\fp(G_\ell) \le \ell$.\label{fi:fanplanar-2}}{\includegraphics[page=2]{fanplanar-linear}}
\caption{Illustration for the proof of \cref{thm:fan-planar-linear}.}
\end{figure}

\begin{proof}
Let $G_\ell$ be the graph constructed as follows; cf.~\cref{fi:fanplanar-1}. 
Start with a $K_{3,3}$. 
\emph{Extend} each edge of the $K_{3,3}$
by adding $\ell-1$ disjoint paths of length two between its endpoints, except for two 
independent edges $(u,v)$ and $(w,z)$. 
Add vertices $w'$ and $z'$, edges $\bar w=(w,w')$ and $\bar z=(z,z')$, $\ell$ disjoint paths of length two 
connecting $w'$ and $z$, and $\ell$ disjoint paths of length two connecting $z'$ and $w$.

Graph $G_\ell$ has $n=6+7(\ell-1)+2+2\ell=9\ell+1$ vertices and admits a drawing with three crossings, see \cref{fi:fanplanar-1}. Recall that we obtain a subdivision of a graph $G$ by subdividing (even multiple times) any subset of its edges. $G_\ell$ contains three subdivisions of $K_{3,3}$ sharing only edge $(u,v)$, and thus each subdivision requires at least one distinct crossing in any drawing. It follows that $\crn(G_\ell)=3$. 
Note that $G_\ell$ admits 
a fan-planar drawing with $\ell$ crossings, cf.\ \cref{fi:fanplanar-2}. 
We prove that $\crn_\fp(G_\ell) = \ell$.
Let $\Gamma$ be a fan-planar drawing of $G_\ell$. If any two extended edges 
cross each other, then the claim follows by \cref{lem:parallel}. Assume they do not:

$G_\ell$ contains $\ell$ subdivions of $K_{3,3}$ that share only $(u,v)$ 
and $\bar w$. Since each $K_{3,3}$ subdivision requires at least one crossing,
there are either $\ell$ crossings in $\Gamma$ (proving the claim), or $(u,v)$ crosses $\bar w$.
Similarly, $G_\ell$ contains~$\ell$ $K_{3,3}$ subdivisions that 
share only $(u,v)$ and $\bar z$, and we can assume that $(u,v)$ crosses $\bar z$. 
But fan-planarity forbids $(u,v)$ to cross both $\bar w$ and $\bar z$.
\end{proof}

\section{Open problems}\label{se:open}

The main open question is whether there exist fan-planar and quasi-planar graphs whose crossing ratio is $\Omega(n^2)$.
In fact, we conjecture that this bound can be reached, but proving our suspected constructions turns out to be elusive.
Another natural research direction is to extend our results to further families of beyond-planar graphs, such as $k$-gap planar graphs or fan-crossing-free graphs (refer to~\cite{DBLP:journals/csur/DidimoLM19} for definitions). 
Finally, we may ask whether similar lower bounds can be proved in the geometric setting (i.e., when the edges are drawn as straight-line segments).

\bibliographystyle{splncs04}
\bibliography{abbrv,paper}

\ifproc
\end{document}
\fi

\clearpage
\appendix

\section{Omitted proofs of \cref{sse:1pl}}\label{app:1pl}

\begin{backInTime}{1-planar-2-drawings}
\begin{lemma}
  \lemOPlanarTDrawings
\end{lemma}

\begin{proof}
Let~\bind{P} be the subgraph of~$G'$ that consists of~$P,f^*$, and the binding edges. 
In any 1-planar drawing of \bind{P},~$f^*$ has to lie in face~$f$ of~$P$:
if we place it in face~$g$, then all binding edges have more than one crossing. 
If we place it in any other face~$h$, then there are at least~$\ell-2$ binding edges that have to leave~$h$, but there are at most~$4<\ell-2$ edges on its boundary, so one of them has to be crossed more than once. Hence, \bind{P} has a unique 1-planar embedding on the sphere.

We now argue that there are only two ways to place the vertices of~$P^*-\{f^*\}$ into the faces of $P$ to obtain a 1-planar embedding of~$G'$.
To this end, observe that there are~$\ell$ disjoint paths (i.e., not sharing any interior vertex) between~$f^*$ and~$g^*$ in~$P^*$. 
We already know that~$f^*$ has to lie in face~$f$.
If we place~$g^*$ in any face~$h$ of~$P$ that has fewer than~$\ell$ edges on its boundary, then each of these disjoint paths has to leave~$h$, so one of the edge on the boundary of~$h$ is crossed more than once.
This leaves only two possibilities for the placement of~$g^*$: either in~$f$ (as depicted in \cref{fi:1-planar-2-crossings}), or in the other polar face~$g$ (as depicted in \cref{fi:1-planar-1-planar}).

We first assume that~$g^*$ lies in~$g$. 
Then, any path from~$g^*$ to~$f^*$ must have at least~7 crossings.
Further, every edge of~$P^*$ lies on a path of length~7 from~$g^*$ to~$f^*$, so every edge of~$P^*$ has to be crossed. 
Since the number of edges in~$P$ and~$P^*$ are equal, also every edge of~$P$ has to be crossed, so there is exactly one vertex of~$P^*$ in every face of~$P$. 
Since~$g^*$ has~$\ell$ neighbors and~$g$ has~$\ell$ edges on its boundary, all its neighbors have to lie in their corresponding face (up to rotation), and by following the disjoint paths, all other vertices of~$P^*$ also have to do so.

Assume now that~$g^*$ lies in~$f$. 
Then we argue that there cannot be a crossing between any edge of~$P$ and an edge of~$P^*$. 
Assume to the contrary that there is one. 
Since~$f^*$ and~$g^*$ lie in face~$f$, there has to be a crossing involving an edge on the boundary of~$f$. 
Let~$(u,v)$ be edge of~$P^*$ involved in this crossing such that~$u$ lies in~$f$ and~$v$ does not.
Then~$v$ lies in a face~$h$ of size~3. 
It is easy to see that there are~$k$ disjoint paths between~$f^*$ and any vertex of~$P^*$ of degree~$k$.
Hence, with the same argument as above,~$v$ must have degree~3.
Assume w.l.o.g. that~$v$ is closer to~$f^*$ than to~$g^*$ in~$P^*$ (the other case is symmetric). We have to distinguish three cases.

\begin{figure}[t]
\begin{subfigure}[t]{.47\columnwidth}
  \includegraphics[page=10]{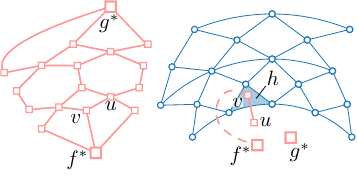}
  \caption{$v$ is adjacent to~$f^*$, but $u\neq f^*$}
  \label{fi:1-planar-u-notf}
\end{subfigure}
\hfill
\begin{subfigure}[t]{.47\columnwidth}
  \includegraphics[page=13]{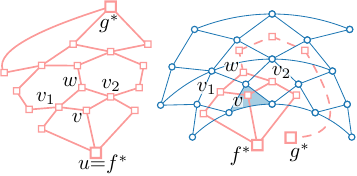}
  \caption{$u=f^*$}
  \label{fi:1-planar-u-f}
\end{subfigure}

\medskip

\begin{subfigure}[t]{.47\columnwidth}
  \includegraphics[page=16]{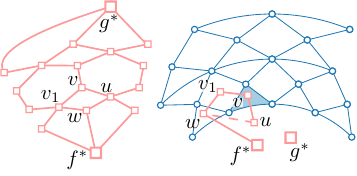}
  \caption{$u$ has degree~4}
  \label{fi:1-planar-u-4}
\end{subfigure}
\hfill
\begin{subfigure}[t]{.47\columnwidth}
  \includegraphics[page=19]{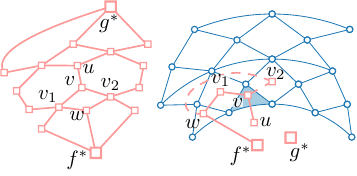}
  \caption{$u$ has degree~3}
  \label{fi:1-planar-u-3}
\end{subfigure}
\caption{Proof that~$P$ and~$P^*$ do not cross when~$g^*$ lies in~$f$.
(Left) situation in~$P^*$ and (Right) proof that the drawing cannot be 1-planar.}
\end{figure}

First, assume that~$v$ is adjacent to~$f^*$. 
If $u\neq f^*$, then the edge $(f^*,v)$ cannot be drawn with one crossing, as the edge between~$f$ and~$h$ is already crossed; see \cref{fi:1-planar-u-notf}.
So $u=f^*$.
Let~$v_1,v_2$ be the other two neighbors of~$v$; see \cref{fi:1-planar-u-f}. 
There are two disjoint paths of length~3 %
in~$P^*$ that go through~$v_1$ and~$v_2$, respectively, and there is exactly one way to draw these two paths in a 1-planar way without crossing the edge between~$f$ and~$h$.
There is a common neighbor~$w$ of~$v_1$ and~$v_2$, so~$w$ has to be placed in a face adjacent to the faces of~$v_1$ and~$v_2$.
However, now there is no way to draw the path from~$w$ to~$g^*$ that consists of~4 edges without multiple crossings.

Second, assume that~$v$ is not adjacent to~$f^*$ and~$u$ has degree~4; see \cref{fi:1-planar-u-4}. 
Let~$v_1$ be the other neighbor of~$v$ of degree~4, and let~$w$ be the common neighbor of~$v_1$ and~$u$. 
The path $(v,v_1,w,f^*)$ has length~3, so it has to be drawn as in \cref{fi:1-planar-u-4} (up to symmetry). 
However, then there is no way to add the edge $(u,w)$ in a 1-planar way.

Finally, assume that~$v$ is not adjacent to~$f^*$ and~$u$ has degree~3; see \cref{fi:1-planar-u-3}. 
Let~$v_1,v_2$ be the other two neighbors of~$v$ and let~$w$ be their common neighbor.
The path $(v,v_1,w,f^*)$ has length~3, so it has to be drawn exactly as in the previous case.
However, then there is no way to draw~$v_2$ such that both~$(v,v_2)$ and~$(w,v_2)$ are crossed only once.

Thus, if~$g^*$ lies in~$f$, then there cannot be any crossing between~$P$ and~$P^*$, so~$P^*$ completely lies inside~$f$.
\end{proof}
\end{backInTime}

\section{A remark to Section~\ref{sse:qp}}

We remark that there is an alternative proof of Theorem~\ref{thm:quasi-planar-linear} and of Corollary~\ref{co:k-quasi-planar} based on the fact that any ($k$)-quasi-planar drawing of $G$ can be redrawn in such a way that (i) the number of crossings is not increased, (ii) special edges are not redrawn, and (iii) each bundle of $\ell$ ``parallel'' paths (an extended edge and the $\ell-1$ paths extending it) has all its $\ell$ paths drawn along almost the same trajectory, thus in particular all the paths in each bundle cross the same set of edges. Such a redrawing is obtained by redrawing, one by one for each bundle, the paths in a bundle by paths drawn along one of them, which has the smallest number of crossings with the edges outside of the bundle. 

\end{document}